\documentclass[11pt]{amsart}

\usepackage{amssymb,amsthm,amsmath,enumerate}
\usepackage[numbers,sort&compress]{natbib}
\usepackage{color}
\usepackage{graphicx}
\usepackage{amssymb,amsthm,amsmath}
\usepackage[numbers,sort&compress]{natbib}
\usepackage{color}
\usepackage{graphicx}

\title[ ]{ Revisiting the Christ-Kiselev's multi-linear operator technique and its applications to Schr\"odinger operators}

\author{Wencai Liu}
\address{ Department of Mathematics, Texas A\&M University, College Station, TX 77843-3368, USA}
\email{liuwencai1226@gmail.com}

\theoremstyle{plain}
\newtheorem{theorem}{Theorem}[section]
\newtheorem{corollary}[theorem]{Corollary}
\newtheorem{lemma}[theorem]{Lemma}
\newtheorem{proposition}[theorem]{Proposition}

\newcommand{\C}{\mathbb{C}}
\newcommand{\R}{\mathbb{R}}
\newcommand{\T}{\mathbb{T}}
\newcommand{\Z}{\mathbb{Z}}

\theoremstyle{definition}

\newtheorem{remark}[theorem]{Remark}

\begin{document}


\begin{abstract}
We established a generalized version of the  Christ-Kiselev's multi-linear operator technique to deal with the spectral theory of Schr\"odinger  operators.
As  applications,  several spectral results of perturbed periodic Schr\"odinger operators are  obtained, including WKB solutions, sharp transitions of  preservation of   absolutely continuous spectra, criteria  of absence of   singular   spectra and sharp bounds of the Hausdorff dimension of singular spectra.
\end{abstract}
\maketitle
\section{Introduction and main results}

Stimulated by several problems posted by Barry Simon at ICMP in 1994 and 2000 \cite{b20}, the   theory of Schr\"odinger operators with decaying potentials have been made  significant progress in the past 25 years through  the work  of    Christ,    Deift,   Denisov,  Killip,  Kiselev,  Last, Molchanov,  Remling, Simon, Stolz and among others.
We refer  readers to two   survey articles   \cite{Kiselevsimon,Killp07} and references therein for details.

  One of the central topics in the area is  to understand the criteria of    preservation of    absolutely continuous spectra  under $L^p$ or power decaying perturbations.
  The one-dimensional case   of free Schr\"odginer operators   is now  well understood mainly by two approaches.
 The first one  is to study the spectral theory of  Schr\"odinger operators via establishing the WKB type eigensolutions.
Another celebrated approach  starting  with Deift and Killip
used a completely different idea: sum rules.
The sum rule approach    can handle the critical case $p=2$, which was missing by WKB methods, but
  less understandings   about
 the dynamics.

 The research by the first approach (WKB method)  has  been achieved a great success  by Christ, Kiselev and Remling \cite{mkjams98,ki96,ki98,mkjfa2,mkjfa1,mkcmp01,recmp98}. In \cite{mkjfa2,mkjfa1,mkcmp01}, Christ and Kiselev    developed a way, referred to as multi-linear operator techniques,   to  deal with the WKB type solutions, which turned out to be a   robust approach.  
In this paper, we will establish a more general  version of  the Christ-Kiselev's multi-linear operator technique than that appearing  in \cite{mkjfa2,mkjfa1,mkcmp01} (see Section \ref{CK}).
We added several significant  technical   ingredients in the proof, which allows us to
  strength  several  conclusions with less assumptions.
 For example, we proved the conclusion in  Lemma \ref{LeSbound} as it was appearing   in  \cite{mkjfa1,mkcmp01,kim}  without  lower bound assumptions on second and third derivatives (see \eqref{Gaug131}).
 As  applications, we obtained  several new important results. We will talk about the details later.
 Our proofs are self-contained except for  several well known facts, such as Christ-Kiselev Lemma (Lemma \ref{CKlemma}) and multiplicative bounds (Theorem \ref{thmmul}).
 We expect our  new version will serve as a good resource for readers to understand Christ-Kiselev's machinery and as well as lead to more   applications in the future.

Multi-linear operator techniques to establish  the WKB type eigensolution are based on  writing down the differential equation in an integral form and
seeking a formal series solution. Each individual  term of the series is defined by a multi-integral operator.
The difficulty lies in giving rigorous  definition  of  improper integrals in a suitable topology, showing the convergence of the series in a proper measure space, and  verifying that the formal series solution is  an actual   solution.
The main scheme of our proofs is definitely  developed from Christ-Kiselev.
However,  several important technical improvements  have been added  to Christ-Kiselev's scheme.
The first new ingredient we want to highlight is that we  define the multi-linear operators   as   iterations, which significantly  simplifies the arguments.  This  is very different from Christ-Kiselev's   plan (see the proof of Lemma 4.2 in \cite{mkjfa1}). Our approach  is more natural and will make it much easier to verify that  the formal series solution is  an actual  solution. The price we need to pay is  to show  the existence of a stronger limit  in the definition of multi-integral operators. See Remark \ref{raug281}.
Another new  important improvement we want to mention is that we  will give two  ways to establish the WKB type solution. One way closely follows from the Christ-Kiselev's   approach.  Our new approach  allows us to avoid using  maximal operators and     simplify the   original approach of Christ and Kiselev. See Section \ref{Newproof}. Our new proof is based  on modifications of norms of a family of Banach spaces.

Finally, we want to highlight several  great discoveries  in our applications, which we believe to be of independent interest. We obtain a nice formula of the spectral measure  of eventually periodic operators by  the combination of Floquet theory   with    Weyl  theory.
 See Section \ref{even}. It  possibly   provides    more opportunities  to use approximations   to investigate  the spectral theory  of  perturbed  periodic Schr\"odinger operators.
We also obtain sharp bounds of a family of oscillatory  integrations, which  are  originally arising from  eigen-equations.
  A simple case of oscillatory  integrations here comes from studying Schr\"odinger operators with    Wigner-von Neumann type potentials \cite{von1929uber}.  Wigner-von Neumann type  functions are  fundamental generators to construct (finitely many or countably many)embedded eigenvalues dating back to Simon \cite{simdense} and receive continuing attentions \cite{ld,luk13,luk14,jlgafa,Kiselev05,liujfa}.  Our sharp bounds  of oscillatory  integrations are   generalizations and quantitative versions of all previous statements in the area  \cite{simdense,liudis,ld,jlgafa,Kiselev05,liujfa,liuasym}.
The novelty here is to split    frequencies into higher and lower ones so that different tools can be used. See  explanations after Corollary \ref{Cor11}.

 Let us move to   applications.  We will investigate
the spectral theory of  the one-dimensional  perturbed periodic  Schr\"odinger  operator, namely,
\begin{equation}\label{Gu}
    Hu=-u^{\prime\prime}+(V(x)+V_0(x))u,
\end{equation}
where $V_0(x)$ is 1-periodic and $V (x)$ is a decaying  perturbation.

When $V\equiv 0$, we have  a  $1$-periodic Schr\"odinger operator,
\begin{equation}\label{GV}
    H_0u=-u^{\prime\prime}+V_0u.
\end{equation}

We are interested in  spectral transitions of operators \eqref{Gu} under decaying perturbations $V$.
The sharp transition for (dense) embedded eigenvalues  was recently obtained by the author and Ong \cite{ld} with some partial results in the past \cite{lotams,KRS,ns,kn,rofe64}.
Sharp transitions of  preservation of   absolutely continuous spectra and  sharp bounds of Hausdorff dimensions of  singular spectra  will be  obtained
as applications of Theorem \ref{prop41}.

 For simplicity, we only consider  the equation on the half line $\R^+$. All the results can be generalized to the whole line $\R$.

Denote by $\ell^p(L^1)(\R^+)$ the Banach space of all measurable functions from $\R^+$ to $\R$ with the norm
\begin{equation*}
  ||f||_{\ell^p(L^1)}=\left(\sum_{k=0}^{\infty}\left(\int_{k}^{k+1}|f(x)|dx\right)^p\right)^{1/p}.
\end{equation*}

This Banach space contains $L^1+L^p$. If $p\leq q$, then $\ell^p(L^1)\subset\ell^q(L^1)$.
 For simplicity, $V_0$ is always assumed to be  in $L^1[0,1]$ and periodic.

Let $S=\cup_{n=0} ^{\infty}(a_{n},b_{n})$ be
the band spectrum of  the operator given by \eqref{GV} and  $\varphi(x, E)$ be the  Floquet solution for $E\in S$.
By Theorem \ref{prop41} and additional spectral analysis, we   have  following two theorems.
\begin{theorem}\label{mainthm1} If the potential $V\in \ell^p(L^1)$  for some $1\leq p<2$, then the set $S$ is an essential support of the absolutely
continuous  spectrum  of the operator $H=H_0+V$ with any boundary condition at zero.
Moreover,
 for a.e.\ $E \in S,$ there exists solution
 $u(x,E)$   of the equation
\begin{equation}
-u''+(V_0(x)+V(x))u = E u
\end{equation}
with the asymptotical behavior
\begin{equation}\label{Gwkb}
 u(x,E)= \varphi(x,E)\exp \left(
\frac{i}{2\Im(\varphi\overline{\varphi}')}
\int\limits_{0}^{x}V(t)|\varphi^{2}(t,E)|\,dt\right)(1+o(1))
\end{equation}
as $x\to\infty$.
\end{theorem}
Let $p^{\prime}$ be the conjugate number to $p$, namely  $\frac{p}{p-1}$ for $p>1$.
\begin{theorem} \label{mainthm2}Suppose  the potential $|x|^{\gamma}V\in \ell^p(L^1)$  for some $1< p\leq2$, $\gamma>0$ with $\gamma p^{\prime}\leq 1$.
Then
 for every \ $E \in S,$ there exists solution
 $u(x,E)$  of $Hu=Eu$ satisfying the asymptotical behavior \eqref{Gwkb}, except for a set of values of $E$ in $S$ with Hausdorff dimension less or equal than $1-\gamma p^{\prime}$.
\end{theorem}
As a corollary, we have
\begin{corollary}\label{Cor}
Suppose $V(x)=\frac{O(1)}{1+x^{\alpha}}$ with $\alpha\in[0,1]$. Then
 for every \ $E \in S,$ there exists solution
 $u(x,E)$  of $Hu=Eu$ satisfying the asymptotical behavior \eqref{Gwkb}, except for a set of values of $E$ in $S$ with Hausdorff dimension less or equal than $2(1-\alpha)$.
\end{corollary}
\begin{remark}
It has been  shown that the multilinear operator  technique is not extendable to  tackle the case $p=2$ by  Muscalu-Tao-Thiele \cite{tao}.
\end{remark}
If $V_0\equiv0$,
Theorems \ref{mainthm1}, \ref{mainthm2} and Corollary \ref{Cor} have been proved   in   \cite{mkcmp01,mkjfa1,mkjfa2,Rempams2000}.
If $V_0\equiv0$, Remling \cite{Remlingsharp}  and Kriecherbauer-Remling \cite{KR01} constructed examples which show that $2(1-\alpha)$ in Corollary \ref{Cor} is the best bound to be achieved.
Under a stronger assumption of  potentials $V$, that is $|x|^\gamma V\in \ell^p(L^1)$  for   some $\gamma>0 $, Theorem \ref{mainthm1} has been proved by Christ and Kiselev
\cite{mkjams98}.  For $p>2$, it is known that Theorem \ref{mainthm1} is not true even for the case $V_0\equiv0$ (see \cite{KLS} for example).
For $p=2$, the first part of Theorem \ref{mainthm1} is proved by  Deift-Killip  for the case $V_0\equiv0$ \cite{deift99} and  Killip \cite{killip02} by different approaches.  For  $p=2$,  the second part is  open  even for $V_0\equiv0$.


Our another interest in this paper is to investigate   spectral  transitions for  singular spectra.
Let us review the results for  the case $V_0\equiv0$ first.
\begin{description}
  \item[I] If $ V(x)=\frac{o(1)}{1+x}$, $H_0+V$ does not have any positive eigenvalues \cite{kato}.
  \item [II]Wigner-von Neumann  type functions imply that  there exist potentials $V(x)=\frac{O(1)}{1+x}$ such that $H_0+V$  has   positive eigenvalues \cite{von1929uber}.
  \item [III] For any  given   positive function  $h(x)$ tending to infinity as  $x\to \infty$,
there exist potentials $V (x) $ such that $|V(x)|\leq \frac{h(x)}{1+x}$
 and  operators $H_0+V$  have dense embedded eigenvalues \cite{nabdense,simdense}.
  \item [IV] If $V(x)=\frac{O(1)}{1+x}$, $H_0+V$ does not have the singular continuous spectrum \cite{Kiselev05}.
  \item [V] For any  given any positive function  $h(x)$ tending to infinity as  $x\to \infty$,
there exist potentials $V (x) $ such that $|V(x)|\leq \frac{h(x)}{1+x}$
 and  singular continuous spectra of  operators $H_0+V$ are non-empty \cite{Kiselev05}.
\end{description}
Clearly, the above statements from I to V imply the criteria  for  the absence of singular spectra (eigenvalues and singular continuous spectra).
It is natural to  expect that  corresponding  criteria  are true for any  non-zero periodic function $V_0$.
For  embedded  eigenvalues, cases I, II and III    have been  proved  to be true for general $V_0$ \cite{KRS,ld}. For the singular continuous spectrum,  we  conjecture  that IV and V  hold for general periodic functions $V_0$.
In this paper, we prove half of the conjecture.

\begin{theorem}\label{mainthm3}
Suppose $V(x)=\frac{O(1)}{1+x}$. Then
the singular continuous spectrum of $H=H_0+V$ with any boundary condition is empty.
\end{theorem}
 Theorem \ref{mainthm3} and case I for general $V_0$  imply
\begin{corollary}
Suppose $V(x)=\frac{o(1)}{1+x}$. Then the spectral measure of $ H_0+V$ with any boundary condition at zero is purely absolutely continuous in $S$.
\end{corollary}
Define $P$ as
\begin{equation}\label{Gdef.Pc}
  P=\{E\in \R:-u^{\prime\prime}+(V(x)+V_0(x))u=Eu \text { has an }L^2(\R^+) \text { solution }\},
\end{equation}
It has been proved that   $P\cap S$ is a countable set provided $V(x)=\frac{O(1)}{1+x}$ \cite{liuasym}.
 Therefore, Theorem \ref{mainthm3} also implies
\begin{corollary}\label{Cor11}
Suppose $V(x)=\frac{O(1)}{1+x}$. Then except for countably many  boundary conditions at zero, the spectral measure of $ H_0+V$  is purely absolutely continuous in $S$.
\end{corollary}
The proof of Theorem \ref{mainthm3} is inspired by \cite{Kiselev05}, where the case $V_0\equiv 0$ has been proved.
Under the assumption of Theorem \ref{mainthm3}, the  singular component of the spectral measure is supported on a  set of zero Hausdoff   dimension  by Corollary \ref{Cor}.
Following the strategy   in \cite{Kiselev05}, additional four steps are needed.  Step 1:  establish the  quantitatively almost orthogonality among   Pr\"ufer angles.
Step 2:  control  the total number of ``separate energies''\footnote{See the   definition in Section \ref{lsection}.} based on step 1.  Step 3:  establish the spectral measure  of   Schr\"odinger operators with eventually zero potentials. Step 4:  use the spectral measure of  eventually zero potentials  to do  approximations.

In our case, the periodic potentials are involved in so that the problem  is a lot more complicated,  in particular, the step 1 and step 3.
Let us mention that the almost orthogonality is between    $\frac{\theta(x,E_1)}{1+x}$ and $\frac{\theta(x,E_2)}{1+x}$ in Hilbert space $ L^2([0,B],(1+x)dx)$ for  large $B$, where $\theta(x,E_1)$ (resp. $\theta(x,E_2)$) is the (generalized) Pr\"ufer angle with respect to energy $E_1$ (resp. $E_2$). In  \cite{Kiselev05}, Kiselev established   sharp bounds of the almost orthogonality of perturbed free Schr\"odinger operators (step 1). For our cases,
 rather than using the standard Pr\"ufer variables, we have to instead use the generalized Pr\"ufer variables, which  is known to be difficult to handle.
 The almost orthogonality of general cases  was proved  very recently in \cite{ld} without quantitative estimates, which is used to  construct embedded eigenvalues. However, in order to deal with the singular continuous spectrum, the quantitative bounds are essential, in particular, we need to control the  blowup when $E_1$ approaches   to $E_2$.
 In \cite{ld},
 one of the  innovations    is the use of Fourier expansions to ensure that some key terms   decay sufficiently quickly.
 For the rest of terms,  it  can be controlled by using    oscillatory integral techniques  to establish  the well cancellation between positive and  negative parts of the integrals.   Even through we use   Fourier expansions and oscillatory integral techniques from \cite{ld} in a quantitative  way, the bounds are not enough.  We overcome the difficulty by splitting  the frequencies into high     and  low  ones, where the frequency comes from  the  quasimomentum of  Floquet theory.
 For  high  frequencies, we quantify   Fourier expansions and oscillatory integral techniques  in \cite{ld}  in a sharp  way.
 For low   frequencies, we  combine  Fourier   expansions in \cite{ld}  with the techniques in  \cite{Kiselev05} to establish the sharp bounds.

In the end, we remark that the spectral theory of perturbed periodic operators in higher dimensions is much more difficult.  We refer readers to  \cite{kuc2016} for details.
\section{ Christ-Kiselev's multi-linear operator techniques}\label{CK}
Since we only consider   operators on the half line $\R_+$, all the functions    are defined on $\R_+$.
Let  us
introduce the  multilinear operator $M_n$, acting on $n$ functions $g_k$, $k=1,2,\cdots,n$,  by
\begin{eqnarray}
  M_n(g_1,g_2,\cdots,g_n)(x,x^{\prime})&=& \int_{x\leq t_1\leq  \cdots\leq t_n\leq x^{\prime}} \prod_{k=1}^n g_k(t_k)dt_k \\
   &=& \int_{x\leq t_1\leq  \cdots\leq t_n<\infty} \prod_{k=1}^n g_k(t_k)\chi_{[0,x']}(t_k)dt_k,\label{Gaug272}
\end{eqnarray}
where $\chi$ is the characteristic function.
If there is a single function $g$ such that $g_k$ is in $\{g,\overline{g}\}$, we write it down by $M_n(g)(x,x^{\prime})$.

A collection of subintervals $E_j^m\subset \R_+$, $1\leq j\leq 2^m$ and $m\in \Z^+$ \cite{mkjfa1} is called a martingale structure if  the following  is true:
\begin{itemize}
  \item $\R_+=\cup_j E_j^m $ for every $m$.
  \item $ E_j^m\cap E_{i}^m=\emptyset$ for every $i\neq j$.
  \item If $i<j, x\in E_i^m$ and $x^\prime\in E_{j}^m$, then $x<x^\prime$.
  \item  For every $m$, $E_j^m=E^{m+1}_{2j-1}\cup E^{m+1}_{2j}$.
\end{itemize}
Denote by $\chi_j^m=\chi_{E_j^m}$.
Let $\mathfrak{B}_s$ be the Banach space consisting of all complex-valued sequences $a=a(m,j)$ indexed by $1\leq m<\infty$ and $1\leq j\leq 2^m$, for which
\begin{equation*}
  ||a||_{\mathfrak{B}_s}=\sum_{m\in \Z_+}m^s\left(\sum_{j=1}^{2^m}|a(m,j)|^2\right)^{1/2}<\infty.
\end{equation*}
Denote by $\mathfrak{B}=\mathfrak{B}_1$.
For any function $g$ on $\R_+$, we can define a sequence  with index $m\in\Z_+$ and $1\leq j\leq 2^m$,
\begin{equation*}
  \left\{\int_{E_j^m}g(x)dx\right\}=\left\{\int_{\R_+} g(x)\chi_j^m dx\right\}.
\end{equation*}
By abusing the notation, denote by
\begin{equation}\label{Gnormb}
||g||_{\mathfrak{B}^s}=\left\|\left\{\int_{E_j^m}g(x)dx\right\}\right\|_{\mathfrak{B}^s}=\sum_{m}^{\infty}m^s\left(\sum_{j=1}^{2^m}\left|\int_{E_j^m}g(x)dx\right|^2\right)^{\frac{1}{2}}.
\end{equation}
Define
\begin{equation}\label{GdefMstar}
  M_n^{*}(g_1,g_2,\cdots,g_n)=\sup_{0<x\leq x^{\prime}<\infty}|M_n(g_1,g_2,\cdots,g_n)(x,x^{\prime})|,
\end{equation}
and
\begin{equation}\label{GdefMstarnew}
  M_n^{*}(g)=\sup_{0<x\leq x^{\prime}<\infty}|M_n(g)(x,x^{\prime})|.
\end{equation}

\begin{theorem}\cite{mkjfa2}\label{thmmul}
For any martingale structure  $E_j^m\subset \R_+$, $1\leq j\leq 2^m$ and $m\in \Z^+$,
the following estimates hold,
\begin{equation}\label{Gaug233}
  M_n^{*}(g_1,g_2,\cdots, g_n)\leq C^{n} \prod_{i=1}^n||g_i||_{\mathfrak{B}},
\end{equation}
and
\begin{equation}\label{Gaug234}
  M_n^{*}(g)\leq C^{n} \frac{||g||_{\mathfrak{B}}^n}{\sqrt{ n!}},
\end{equation}
where $C$ is an absolute  constant.
\end{theorem}
A martingale structure  $E_j^m\subset \R_+$, $1\leq j\leq 2^m$ and $m\in \Z^+$ is said to be adapted in $\ell^p(L^1)$ to a function $f$ if for all possible $m,j$,
\begin{equation}\label{Gaug2310}
  ||f\chi_j^m||^p_{\ell^p(L^1)}\leq  2^{-m}||f||^p_{\ell^p(L^1)}.
\end{equation}
Since all the functions are in  $\ell^p(L^1)$, we omit ``adapted" in the rest of this paper.
\begin{lemma}(p.433, \cite{mkjfa1})
For any function $f\in\ell^p(L^1)$,  there exists a martingale structure  $\{E_j^m\subset \R_+:m\in \Z_+,  1\leq j\leq 2^m\}$ to $f$.
\end{lemma}

 Suppose $P $ is a  linear  or sublinear operator: for any function $f$ on $\R_+$, $P(f)(\lambda)$ is a function on a interval $J\subset \R$.
For $s>0$, denote by
\begin{eqnarray*}
   G_{P(f)(\lambda)}^{(s)} &=& ||\{P (f\chi_j^m)(\lambda)\}||_{\mathfrak{B}^s} \\
  &=& \sum_{m=1}^{\infty}m^s\left(\sum_{j=1}^{2^m} |P(f\chi_j^m)(\lambda)|\right)^{1/2}.
\end{eqnarray*}

\begin{remark}
 In the case that $P$ has an integral kernel $p(x,\lambda)$,  $G^{(s)}_{P(f)(\lambda)}=||p(x,\lambda)f(x)||_{\mathfrak{B}^s}$, that  is the norm of $\{\int_{E_j^m} p(x,\lambda)g(x)dx\}$ in $\mathfrak{B}^s$.

\end{remark}
The following statement is  from  \cite{mkjfa1}. We include a proof here for completeness.

\begin{theorem}\cite[Prop.3.3]{mkjfa1}\label{thmbound}
Given a function $f\in  \ell^p(L^1)$,  fix a  martingale structure to $f$.
 Suppose $P$ is a bounded linear  or sublinear operator
from $ \ell^p(L^1)$ to $L^q(J)$, where $1\leq p<2<q$  and $J\subset \R$ is a closed interval.
Then
\begin{equation*}
  ||  G^{(s)}_{P(f)(\lambda)}||_{L^{q}(J)} \leq C(p,q,s, ||P||) ||f||_{\ell^p(L^1)}.
\end{equation*}
\end{theorem}
\begin{proof}
Let
\begin{equation*}
  t_m(\lambda)=  \left(\sum_{j=1}^{2^m}|P (f\chi_j^m)(\lambda)|\right)^{1/2}.
\end{equation*}
By the definition,
\begin{eqnarray}
  G^{(s)}_{P(f)(\lambda)} &=& \sum_{m=1}^{\infty}m^s\left(\sum_{j=1}^{2^m} |P (f\chi_j^m)(\lambda)|\right)^{1/2}\nonumber\\
   &=& \sum_{m=1}^{\infty}m^st_m(\lambda). \label{Gaug2311}
\end{eqnarray}
Let us give an inequality, for $\gamma\geq1$
\begin{equation}\label{Ginequa1}
  \left(\sum_{i=1}^N a_i\right)^{\gamma}\leq N^{\gamma-1}\sum_{i=1}^N |a_i|^{\gamma}.
\end{equation}
Direct computations imply
\begin{eqnarray}
\int_{J}t_m^q(\lambda)d\lambda&=&  \int_J \left(\sum_{j=1}^{2^m}|P(f\chi_j^m)(\lambda)|^2\right)^{q/2} d\lambda\nonumber  \\
   &\overset{\text { by }\eqref{Ginequa1}}{\leq}&
    2^{m(q/2-1)}\int_J \sum_{j=1}^{2^m}| P(f\chi_j^m)(\lambda)|^qd\lambda\nonumber  \\
    &\overset{\star}{\leq}&C
    2^{m(q/2-1)} \sum_{j=1}^{2^m}  ||f\chi_j^m||_{\ell^{p}(L^1)}^q\nonumber  \\
   &\overset{\text { by }\eqref{Gaug2310}}{\leq}& C  2^{m(q/2-1)} \sum_{j=1}^{2^m} 2^{-m\frac{q}{p}}||f||_{\ell^p(L^1)}^q\nonumber \\
    &\leq& C ||f||_{\ell^p(L^1)}^q2^{m\frac{q}{2}-m\frac{q}{p}},\label{Gaug2312}
\end{eqnarray}
where the $\star $ holds by the boundedness of $P$.

Finally, we have
\begin{eqnarray*}
 || G^{(s)}_{S(f)(\lambda)}||_{L^q(J)}&\overset{\text { by }\eqref{Gaug2311}}{=}&||\sum_{m=1}^{\infty}m^st_m(\lambda)||_{L^q(J)}\\
  &\leq& \sum_{m=1}^{\infty}m^s||t_m(\lambda)||_{L^q(J)}\\
 &\overset{\text { by }\eqref{Gaug2312}}{\leq}& C\sum_{m=1}^{\infty}m^s 2^{m/2-m/p}||f||_{\ell^p(L^1)}\\
  &\leq&C||f||_{\ell^p(L^1)}.
\end{eqnarray*}
\end{proof}
Denote by
\begin{equation}\label{Gaug221}
 B_n({g}_1,{g}_2,\cdots,{g}_n)(x)=\int_{x}^{\infty}\int_{t_1}^{\infty}\cdots\int_{t_{n-1}}^{\infty}\prod_{j=1}^n g_j(t_j)dt_1dt_2\cdots dt_n.
\end{equation}

If there is a single function $g$ such that $g_k= \{g,\bar{g}\}$, $k=1,2,\cdots,n$, we write it down by $B_n(g)(x)$.
\begin{theorem}\label{prop41}
Assume $g_j$ is locally integrable, $j=1,2,\cdots,n$.
Suppose for $j=1,2,\cdots, n$,
\begin{equation}\label{Gaug273}
  \limsup_{M\to \infty}||g_j\chi_{[M,\infty)}||_{\mathfrak{B}}=0,
\end{equation}
and there is a constant $C$ (does not depend on $I$)  such that for any closed interval  $I\subset \R_{+}$,
\begin{equation}\label{Gaug254}
||g_j\chi_{I}||_{\mathfrak{B}} \leq C.
\end{equation}
Then
 \eqref{Gaug221}  is well defined as the  limits
\begin{equation}\label{Gaug232}
  B_n(g_1,g_2,\cdots,g_n)(x)=\lim_{y_1,\cdots,y_n\to \infty}\int_{x}^{y_1}\int_{t_1}^{y_2}\cdots\int_{t_{n-1}}^{y_n}\prod_{j=1}^n g_j(t_j)dt_1dt_2\cdots dt_n,
\end{equation}
and
\begin{equation}\label{Gaug256}
 \lim_{x\to\infty}  B_n(g_1,g_2,\cdots,g_n)(x)=0.
\end{equation}
Moreover, for almost every $x$
\begin{equation}\label{Gaug256new}
  \frac{dB_n(g_1,g_2,\cdots,g_n)(x)}{dx}= -g_1(x)B_{n-1}(g_2,\cdots,g_n)(x).
\end{equation}
\end{theorem}
\begin{proof}
Without loss of generality, assume $n=2$.
In order to prove the existence of the limits,
it suffices to show  that
\begin{equation}\label{Gaug235}
  \lim_{y_1,z_1,y_2,z_2\to \infty}\sup_{x}|B_2(g_1\chi_{[0,y_1]},g_2\chi_{[0,y_2]})(x)-B_2(g_1\chi_{[0,z_1]},g_2\chi_{[0,z_2]})(x)|=0.
\end{equation}
Assume $M<y_1<z_1$ and $M<y_2<z_2$.
By telescoping techniques,
\begin{equation*}
 |B_2(g_1\chi_{[0,y_1]},g_2\chi_{[0,y_2]})(x)-B_2(g_1\chi_{[0,z_1]},g_2\chi_{[0,z_2]})(x)|\;\;\;\;\;\;\;\;\;\;\;\;\;\;\;\;\;\;\;\;\;\;\;\;\;\;\;\;\;\;\;\;\;\;\;\;\;\;\;\;
\end{equation*}
\begin{eqnarray}
   &\leq&  |B_2(g_1\chi_{[0,y_1]},g_2\chi_{[y_2,z_2]})(x)|+|B_2(g_1\chi_{[y_1,z_1]},g_2\chi_{[0,z_2]})(x)|\nonumber\\
   &\overset{\text { by }\eqref{Gaug233} \text{ and } \eqref{Gaug272}}{\leq} &  C ||g_1\chi_{[0,y_1]}||_{\mathfrak{B}} ||g_2\chi_{[y_2,z_2]}||_{\mathfrak{B}}+C||g_1\chi_{[y_1,z_1]}||_{\mathfrak{B}}||g_2\chi_{[0,z_2]}||_{\mathfrak{B}}\nonumber\\
     &\overset{\text { by }\eqref{Gaug254}}{\leq}& C  ||g_2\chi_{[y_2,z_2]}||_{\mathfrak{B}}+C||g_1\chi_{[y_1,z_1]}||_{\mathfrak{B}}\nonumber\\
      &\leq& C  ||g_2\chi_{[M,\infty)}||_{\mathfrak{B}}+C||g_1\chi_{[M,\infty)}||_{\mathfrak{B}}.\label{Gaug251}
\end{eqnarray}
Now  \eqref{Gaug235} follows from \eqref{Gaug273}.

By \eqref{Gaug232}, one has
\begin{eqnarray*}
 \lim_{x\to \infty} |B_2(g_1,g_2)(x) |&=&\lim_{x\to\infty} \lim_{x'\to\infty}\left|\int_{x}^{x'}\int_{t_1}^{x'}\prod_{j=1}^2 g_j(t_j)dt_1dt_2\right|  \\
   &=& \lim_{x\to\infty} \lim_{x'\to\infty}
   \left|\int_{x}^{x'}\int_{t_1}^{x'}\prod_{j=1}^2 g_j(t_j)\chi_{[x,\infty)}(t_j)dt_1dt_2\right|\\
    &\overset{\text { by }\eqref{Gaug233}}{\leq}&C\lim_{x\to\infty}||g_1\chi_{[x,\infty)}||_{\mathfrak{B}} ||g_2\chi_{[x,\infty)}||_{\mathfrak{B}}\\
    &\overset{\text { by }\eqref{Gaug273}}{=}&0.
\end{eqnarray*}
This completes the proof of \eqref{Gaug256}.
 Direct computations imply
\begin{eqnarray*}
         \lim_{y\to x-}   \frac{ B_2(g_1,g_2)(y)-B_2(g_1,g_2)(x)}{y-x} &=& \lim_{y\to x-}\frac{1}{y-x}\int_{y}^{x} g_1(t_1)dt_1\int_{ t_1}^{\infty}g_2(t_2)dt_2\\
                &=&- g_1(x)\int_{ x}^{\infty}g_2(t_2)dt_2.
             \end{eqnarray*}

Similarly,$$\lim_{y\to x+}   \frac{ B_2(g_1,g_2)(y)-B_2(g_1,g_2)(x)}{y-x}=-g_1(x)\int_{ x}^{\infty}g_2(t_2)dt_2.$$
The last two equalities imply \eqref{Gaug256new}.
\end{proof}
\begin{remark}\label{raug281}
In \cite[Prop.4.1]{mkjfa1},  a weaker    limit $ \tilde{B}_n(g_1,g_2,\cdots,g_n)(x)$ was proved to be existed, where
\begin{equation*}
  \tilde{B}_n(g_1,g_2,\cdots,g_n)(x)=\lim_{y\to \infty}\int_{x}^{y}\int_{t_1}^{y}\cdots\int_{t_{n-1}}^{y}\prod_{j=1}^n g_j(t_j)dt_1dt_2\cdots dt_n.
\end{equation*}
\end{remark}
Let $p(x,\lambda)$ be a measurable function on $\R_+\times J$. Define the integral  operator $ P$:
\begin{equation*}
  P(f)(\lambda)=\int_{\R_+}p(x,\lambda)f(x)dx,
\end{equation*}
and the maximal operator  $P^{\ast}$:
\begin{equation*}
  P^{\ast}(f)(\lambda)=\sup_{y\in\R_+}\left|\int_{y}^{\infty}p(x,\lambda)f(x)dx\right|.
\end{equation*}

\begin{lemma}\cite[Christ-Kiselev lemma]{mkjfa2}\label{CKlemma}
Let $1\leq p<q<\infty$.
Suppose $P$ is a bounded operator  from $\ell^p(L^1)$ to $ L^q(J)$.
Then $P^{\ast}$ is also a bounded operator  from $\ell^p(L^1)$ to $ L^q(J)$.
\end{lemma}
In our situation (see next section), $s(x,\lambda)=w(x,\lambda)e^{-ih(x,\lambda)}$, where $h$ is a real-valued function.
We obtain two corresponding  operators
\begin{equation}\label{Gdefs}
   S(f)(\lambda)=\int_{\R_+}w(x,\lambda)e^{-ih(x,\lambda)}f(x)dx,
\end{equation}
and
\begin{equation}\label{Gdefs1}
  S^{\ast}(f)(\lambda)=\sup_{y\in\R_+}\left|\int_{y}^{\infty}w(x,\lambda)e^{-ih(x,\lambda)}f(x)dx\right|.
\end{equation}

\begin{lemma}\label{LeSbound}
Assume $1\leq p\leq 2$.
Suppose there exist a constant $C$  and a  open interval $\tilde{J}$  satisfying $J\subset \tilde{J}$ such that    for $\lambda\in \tilde{J}$
\begin{equation}\label{Gaug131}
|\partial_{\lambda} [ h(x,\lambda)-h(y,\lambda)]|\geq \frac{|x-y|}{C}
\end{equation}
and  for $i=1,2,3$ and $\lambda\in \tilde{J}$
\begin{equation}\label{Gaug132}
|\partial_{\lambda}^{i} [ h(x,\lambda)-h(y,\lambda)]|\leq  C|x-y|
\end{equation}
provided $|x-y|\geq C$.
Suppose
\begin{equation*}
\sup_{x\in \R_+,\lambda\in \tilde{J}}  \sum_{i=1}^2|\partial_{\lambda}^iw(x,\lambda)|\leq C.
\end{equation*}
Let $p^{\prime}=\frac{p}{p-1}$ be the conjugate exponent to $p$. Then
\begin{equation*}
    ||Sf||_{L^{p^{\prime}}(J,d\lambda)}\leq O(1)||f||_{\ell^p(L^1)},
\end{equation*}
and
\begin{equation*}
    ||S^{\ast}f||_{L^{p^{\prime}}(J,d\lambda)}\leq O(1)||f||_{\ell^p(L^1)},
\end{equation*}
where $O(1)$ depends on $C$, $J$, $ \tilde{J}$ and $p$.
\end{lemma}

\begin{proof}
By Lemma \ref{CKlemma}, we only need to prove the boundedness of $S$.
By interpolation theorem, it suffices to prove the cases $p=1$ and $p=2$. The case of $p=1$ is trivial since $h$ is a real-valued function.
So we only need to consider the case $p=2$.
Let $\xi(\lambda)$ be a positive function so that $\xi\equiv1 $ on $J$ and $ {\rm supp} \xi \subset \tilde{J}$.
Then one has
\begin{eqnarray}
  ||Sf||^2_{L^{2}(J,d\lambda)}&=& \int_J\left |\int_{\R_+} w(x,\lambda)e^{-ih(x,\lambda)}f(x)dx\right|^2d \lambda\nonumber\\
 &\leq& \int_{\tilde{J}} \left|\int_{\R_+}w(x,\lambda) e^{-ih(x,\lambda)}f(x)dx\right|^2\xi(\lambda)d \lambda\nonumber\\
 &=& \int_ {\tilde{J} }\left[\int_{\R_+} w(x,\lambda)e^{-ih(x,\lambda)}f(x)dx\right]\left[    \int_{\R_+} \bar{w}(y,\lambda)e^{ih(y,\lambda)}\bar{f}(y)dy\right]\xi(\lambda)d \lambda \nonumber\\
 &=&\int_{\R^2_+}f(x)\bar{f}(y)dxdy \int_{ \tilde{J}} e^{-ih(x,\lambda)+ih(y,\lambda)}w(x,\lambda)\bar{w}(y,\lambda)\xi(\lambda)d \lambda.\label{Gaug193}
\end{eqnarray}
Multiplying $-i\partial_{\lambda}(h(x,\lambda)-h(y,\lambda)) $ and dividing $-i\partial_{\lambda}(h(x,\lambda)-h(y,\lambda)) $ and integrating  by part twice, we have for $|x-y|\geq C$,
\begin{equation*}
   \int_ {\tilde{J}} e^{-ih(x,\lambda)+ih(y,\lambda)}w(x,\lambda)\bar{w}(y,\lambda)\xi(\lambda)d \lambda\;\;\;\;\;\;\;\;\;\;\;\;\;\;\;\;\;\;\;\;\;\;\;\;\;\;\;\;\;\;\;\;\;\;\;\;\;\;\;\;\;\;\;\;\;\;\;\;\;\;\;\;\;\;\;\;\;\;\;\;
\end{equation*}
\begin{align}
\;\;\;\;\;\;\;\;=&\int_{ \tilde{J}} \frac{-i\partial_{\lambda}(h(x,\lambda)-h(y,\lambda))}{-i\partial_{\lambda}(h(x,\lambda)-h(y,\lambda))} e^{-ih(x,\lambda)+ih(y,\lambda)}w(x,\lambda)\bar{w}(y,\lambda)\xi(\lambda)d \lambda \nonumber\\
 =&\int_{ \tilde{J}} e^{-ih(x,\lambda)+ih(y,\lambda)}\partial_{ \lambda}\left(\frac{w(x,\lambda)\bar{w}(y,\lambda)\xi(\lambda)}{-i\partial_{\lambda}(h(x,\lambda)-h(y,\lambda))}\right)d\lambda\nonumber\\
 =& \int_{ \tilde{J}} e^{-ih(x,\lambda)+ih(y,\lambda)}\partial_{\lambda}\left[\frac{1}{-i\partial_{\lambda}(h(x,\lambda)-h(y,\lambda))}\partial_{\lambda}\left(\frac{w(x,\lambda)
 \bar{w}(y,\lambda)\xi(\lambda)}{-i\partial_{\lambda}
 (h(x,\lambda)-h(y,\lambda))}\right) \right] d\lambda\nonumber\\
 =&\frac{O(1)}{|x-y|^2}\label{Gaug194},
\end{align}
where the last  equality holds by Lemma \ref{Lemontoneh}.

By \eqref{Gaug193} and \eqref{Gaug194}, we have
\begin{eqnarray}
  ||Sf||^2_{L^{2}(J,d\lambda)}&\leq& \int_{|x-y|>C}f(x)\bar{f}(y)dxdy \int_{ \tilde{J}} e^{-ih(x,\lambda)+ih(y,\lambda)}w(x,\lambda)\bar{w}(y,\lambda)\xi(\lambda)d \lambda. \nonumber\\
  && +\int_{|x-y|\leq C}f(x)\bar{f}(y)dxdy \int_{ \tilde{J}} e^{-ih(x,\lambda)+ih(y,\lambda)}w(x,\lambda)\bar{w}(y,\lambda)\xi(\lambda)d \lambda\nonumber\\
 &=&O(1)\int_{\R_+^2}\frac{|f(x)f(y)|}{1+|x-y|^2}dxdy\nonumber \\
&=&O(1)||f||_{\ell^2(L^1)},\label{Gaug195}
\end{eqnarray}
where the the last equality holds by   calculations (for convenience, we include  the details in the appendix).
This completes the proof.
\end{proof}
\begin{remark}
The formulation and  proof of Lemma \ref{LeSbound}  closely  follows from  the corresponding parts appearing  in \cite{mkjfa1,mkcmp01,kim}.

\end{remark}
\section{Preparations for the applications }
In the rest of this paper, we will apply Theorem \ref{prop41} to prove  Theorems \ref{mainthm1}, \ref{mainthm2} and \ref{mainthm3}.
Let us set up the basics in this section and do the proof in the following sections.
Since $\varphi(x,E)$ is the Floquet solution for $E\in \cup (a_n,b_n)$, one has
\begin{equation}\label{Gdefj}
    \varphi(x,E)=J(x,E)e^{ik(E)x},
\end{equation}
where $k(E)\in[0,\pi]$ is the quasimomentum, and $J(x,E)$ is 1-periodic.
Since $\varphi(x,E)$ and $\overline{\varphi}(x,E)$ are two linearly independent solutions of $u^{\prime\prime}+V_0u=Eu$, the
Wronskian $W(\overline{\varphi},\varphi)$ is constant and
\begin{equation}\label{Gwron}
  W(\overline{\varphi},\varphi)=\overline{\varphi}(x)\varphi^{\prime}(x)-\overline{\varphi}^{\prime}(x)\varphi(x)=2i\Im [\overline{\varphi}(x)\varphi^{\prime}(x)].
\end{equation}
By interchanging $\overline{\varphi}$ and $\varphi$, we always assume
\begin{equation}\label{Gflight1}
    W(\overline{\varphi},\varphi)=i\omega,
\end{equation}
with $\omega>0$.

Let us study the solutions of the equation
\[ -u''+(V_0(x)+V(x))u=Eu. \]
We rewrite this
equation as a linear system
\[ u_{1}' = \left( \begin{array}{cc} 0 & 1 \\ V_0+V-E &
 0 \end{array} \right)u_{1}, \]
where $u_{1}$ is   a vector $\left(
 \begin{array}{c} u\\ u' \end{array} \right).$
Introduce
\[ u_{1} = \left( \begin{array}{cc} \varphi (x, E) & \overline{\varphi}
 (x, E)
\\ \varphi '(x, E) & \overline{\varphi}' (x, E) \end{array}
\right) u_{2}. \]
Then
\begin{equation} u_{2}' = \frac{i}{2\Im (\varphi \overline{\varphi}')}\left(
\begin{array}{cc} V(x)|\varphi (x, E)|^{2} & V(x) \overline{\varphi}
 (x, E)^{2} \\ -V(x) \varphi (x, E)^{2} & -V(x) |\varphi (x,
E)|^{2} \end{array} \right)u_{2}.
\end{equation}
Define
\[ p(x,E) = \frac{1}{2\Im (\varphi \overline{\varphi'})}
\int\limits_{0}^{x}V(y) |\varphi(y,E)|^{2}\,dy. \]
Let us apply  another   transformation,
\[ u_{2} = \left( \begin{array}{cc} \exp (ip(x,E))  & 0
 \\ 0 & \exp (-ip(x,E)) \end{array} \right)u_{3} . \]
We obtain the    equation for $u_{3}:$
\begin{equation}\label{Gu_3}
u_{3}' = \frac{i}{2\Im (\varphi \overline{\varphi}')}\left(
\begin{array}{cc} 0 & V(x) \overline{\varphi}
 (x, E)^{2}\exp (-2ip(x,E)) \\ -V(x) \varphi (x, E)^{2}
\exp (2ip(x,E)) & 0 \end{array} \right)u_{3}.
\end{equation}
\begin{lemma}\label{le10}
Suppose there exists a solution of \eqref{Gu_3} satisfying
\begin{equation*}
u_3(x,E)=
  \left(\begin{array}{c}
    1 \\
     0
   \end{array}
   \right)+o(x)
\end{equation*}
as $x\to \infty$. Then  there exists a  solution $u(x,E)$ of \eqref{Gu} satisfying  \eqref{Gwkb}.
\end{lemma}
\begin{proof}
The proof is straitfoward by   subistitutions.
\end{proof}
Let $Y=u_3$. Let $\phi $ be so that $e^{i\phi}=\varphi(x,E)$. Denote by
\begin{equation}\label{def.f}
 w(x,E)=\frac{i}{2\Re \phi^{\prime}},
\end{equation}
and
\begin{equation}\label{Gdefh}
  h(x,E)=2 \Re \phi-\int_{0}^x \frac{V(t)}{\Re \phi^{\prime}(t,E)}dt.
\end{equation}

In the following,  $ w$ and $ h$ are always given by \eqref{def.f} and \eqref{Gdefh} correspondingly.
The operators $S$ and $S^{\ast}$ are given by \eqref{Gdefs}  and \eqref{Gdefs1} correspondingly.
Denote by
\begin{equation}\label{Gaug275}
 \mathcal{F}(x,E)= { {w}}(x,E) e^{-ih(x,E)}V(x).
\end{equation}
Under those  notations and following  calcaluations in p. 249 and p.250   in \cite{mkcmp01},    \eqref{Gu_3} becomes
    \begin{equation}\label{Gy}
Y' =  \left(
\begin{array}{cc} 0 &  w e^{-ih}V \\ \bar{ {w}} e^{ih}V & 0 \end{array} \right)Y=\left(
\begin{array}{cc} 0 &  \mathcal{F} \\ \bar{  \mathcal{F}} & 0 \end{array} \right)Y.
\end{equation}
For convenience, we include a verification of \eqref{Gy} in the Appendix.

Denote by
\begin{equation*}
D=   \left(
\begin{array}{cc} 0 &  \mathcal{F} \\ \bar{ \mathcal{F}} & 0 \end{array} \right).
\end{equation*}

The linear equation \eqref{Gy} becomes $Y^{\prime}=DY$.
We are going to find a solution as
\begin{equation}\label{Gseri}
  Y(x)=\left(\begin{array}{c}
         1 \\
         0
       \end{array}\right)-\int_{x}^{\infty}D(y)Y(y)dy,
\end{equation}
and we obtain a   series solution by iterations
\begin{equation}\label{Gseri1}
  Y(x)=\left(\begin{array}{c}
         1 \\
         0
       \end{array}\right)+\sum_{k=1}^{\infty}(-1)^k\int\cdots\int_{x\leq t_1\leq t_2\cdots\leq t_k<\infty}D(t_1)D(t_2)\cdots D(t_k)\left(\begin{array}{c}
         1 \\
         0
       \end{array}\right)dt_kdt_{k-1}\cdots dt_2dt_1.
\end{equation}
Let
\begin{equation*}
T_n(\mathcal{F})(x,x',E)= M_{n}(\mathcal{F}(\cdot,E))(x,x').
\end{equation*}
Under those  notations, one has
\begin{equation*}
  \int\cdots\int_{x\leq t_1\leq t_2\cdots\leq t_{2k}\leq x'}D(t_1)D(t_2)\cdots D(t_{2k})\left(\begin{array}{c}
         1 \\
         0
       \end{array}\right)dt_{2k}\cdots dt_2dt_1=\left(\begin{array}{c}
         T_{2k}(\mathcal{F})(x,x',E) \\
         0
       \end{array}\right),
\end{equation*}
and
\begin{equation*}
  \int\cdots\int_{x\leq t_1\leq t_2\cdots\leq t_{2k+1}\leq x'}D(t_1)D(t_2)\cdots D(t_{2k+1})\left(\begin{array}{c}
         1 \\
         0
       \end{array}\right)dt_{2k+1}\cdots dt_2dt_1=\left(\begin{array}{c}
         0 \\
          T_{2k+1}(\mathcal{F})(x,x',E)
       \end{array}\right).
\end{equation*}
The series solution \eqref{Gseri1} becomes
\begin{equation}\label{Gseri2}
  Y(x)=\left(\begin{array}{c}
         1 \\
         0
       \end{array}\right)+\left(\begin{array}{c}
         \sum_{m=1}^{\infty}T_{2m}(\mathcal{F})(x,\infty,E) \\
         -\sum_{m=0}^{\infty}T_{2m+1}(\mathcal{F})(x,\infty,E)
       \end{array}\right).
\end{equation}
We will show \eqref{Gseri1} and \eqref{Gseri2} are well defined and give an  actual solution  of \eqref{Gy}.

\section{Proof of Theorem \ref{mainthm1} }
Fix a martingale structure  $\{E_j^m\subset \R_+:m\in \Z_+,  1\leq j\leq 2^m\}$ to the potential $V$.
Choose a spectral band $(a_n,b_n)$ and let $K\subset (a_n,b_n)$ be an arbitrary closed interval. We will apply  Theorem \ref{prop41} to complete our proof.
In order to make it different from the abstract statements, instead of using parameter $\lambda$, the energy parameter $E$ will be used.
\begin{lemma}\label{Lemontoneh}
For any $E\in (a_n,b_n)$,
there exists a constant $C=C(E)$ (depends on $E$ uniformly in any compact subset of $(a_n,b_n)$) such that
\begin{equation}\label{Gaug131new}
|\partial_{E} [ h(x,E)-h(y,E)]|\geq \frac{|x-y|}{C}
\end{equation}
and  for $i=1,2,3$
\begin{equation}\label{Gaug132new}
|\partial_{E}^{i} [ h(x,E)-h(y,E)]|\leq  C|x-y|
\end{equation}
provided $|x-y|\geq C$.
\end{lemma}
\begin{proof}
We will prove \eqref{Gaug131new} first.
By the definition of $\phi$ and \eqref{Gdefj}, one has
\begin{equation}\label{Gequ2}
    \Re \phi=k(E)x+\Im \log J(x,E),
\end{equation}
with $k(E)\in (0,\pi)$. By the Floquet theory,
\begin{equation}\label{Gaug133}
  \frac{d k(E)}{dE}\neq 0.
\end{equation}
By the fact that $J(x,E)$ is 1-periodic, one has
\begin{equation}\label{Gequ1}
   \Im \log J(x+1,E)- \Im \log J(x,E)=2q\pi,
\end{equation}
for some $q\in \Z$. It implies
\begin{equation}\label{Gaug134}
   \partial_E(\Im \log J(x+1,E)- \Im \log J(x,E))=0.
\end{equation}

By \eqref{Gaug133} and \eqref{Gaug134}, we have
\begin{equation}\label{Gequ4}
    | \Re \phi(x,E)-\Re \phi(y,E)|\geq  \frac{|x-y|}{C} .
\end{equation}
Since   $V(x)\in \ell^p(L^1)$,
one has $\partial_E \int_N^{N+1}\frac{V(t)}{\Re \phi^{\prime}(t,E)}dt$   goes to zero as $N\to\infty$.
It implies
\begin{equation}\label{Gequ3}
\left|\partial_E\int_{x}^y \frac{V(t)}{\Re \phi^{\prime}(t,E)}dt\right|=o(y-x)+O(1),
\end{equation}
as $y-x$ goes to $\infty$.
Now \eqref{Gaug131new} follows from \eqref{Gequ4} and \eqref{Gequ3}.
The proof of \eqref{Gaug132new} can be proceeded in a similar way.
\end{proof}

\begin{lemma}\label{LeGbound}
Let $q$ be the number conjugate to $p$ with $1\leq p<2$.
Then the following estimate holds
\begin{equation*}
 G^{(s)}_{S^{\ast}(V)(E)} \in L^{q}(K,d E).
\end{equation*}
In particular ($s=1$),
\begin{equation*}
 G_{S^{\ast}(V)(E)} \in L^{q}(K,d E).
\end{equation*}
\end{lemma}
\begin{proof}
By Lemmas \ref{Lemontoneh}, \ref{LeSbound} and applying $P=S^{\ast}$ in Theorem \ref{thmbound}, we have
\begin{equation}\label{Gaug244}
G^{(s)}_{S^{\ast}(V)(E)} \in L^{q}(K,d E).
\end{equation}
\end{proof}

\begin{corollary}\label{Cor1}
Let $q$ be the number conjugate to $p$ with $1\leq p<2$.
Then for almost every $E\in K$,
\begin{equation}\label{Gaug277}
||\mathcal{F}(\cdot,E)\chi_{I}||_{\mathfrak{B}} \leq C(E),
\end{equation}
for any closed interval $I$,
and
\begin{equation}\label{Gaug276}
\limsup_{M\to\infty}||\mathcal{F}(\cdot,E)\chi_{[M,\infty)}||_{\mathfrak{B}} =0.
\end{equation}
\end{corollary}
\begin{proof}
Direct computations,  one has
\begin{eqnarray*}
  ||\mathcal{F}(\cdot,E)\chi_{I}||_{\mathfrak{B}} &=& \left\|\left\{\int_{E_j^m}\mathcal{F}(x,E)\chi_{I}(x) dx\right\}\right\|_{\mathfrak{B}} \\
   &=& \left\|\left\{\int_{E_j^m}w(x,E)e^{ih(x,E)}V(x)\chi_{I}(x) dx\right\}\right\|_{\mathfrak{B}}\\
   &=& \left\|\left\{\int_{I}w(x,E)e^{ih(x,E)}V(x)\chi_j^m dx\right\}\right\|_{\mathfrak{B}}\\
     &\leq&2 \left\|\left\{S^{\ast}(V\chi_{j}^m)(E)\right\}\right\|_{\mathfrak{B}}\\
     &=&2 G_{S^{\ast}(V)(E)}.
\end{eqnarray*}
Now \eqref{Gaug277} follows since $ G_{S^{\ast}(V)(E)} \in L^{q}(K,d E)$ by Lemma \ref{LeGbound}.
Applying $P=S$ and $ f=V(x)\chi_{[M,\infty)}$ in Theorem \ref{thmbound} and recalling that
\begin{equation*}
 ||\mathcal{F}(\cdot,E)\chi_{[M,\infty)}||_{\mathfrak{B}}=G_{S(V\chi_{[M,\infty)})},
\end{equation*}
one has
\begin{eqnarray}
 \limsup_{M\to \infty} \left(\int_J||\mathcal{F}(\cdot,E)\chi_{[M,\infty)}||^q_{\mathfrak{B}}dE\right)^{1/q}&=&  \limsup_{M\to \infty} ||G_{S(V\chi_{[M,\infty)})}||_{L^q(J)} \\
 &\leq&  O(1)\limsup_{M\to \infty} ||V \chi_{[M,\infty)}||_{\ell^p(L^1)} \nonumber\\
   &=& 0. \nonumber
\end{eqnarray}
It implies  for almost every $E\in K$,
\begin{equation*}
  \limsup_{M\to\infty}||\mathcal{F}(\cdot,E)\chi_{[M,\infty)}||_{\mathfrak{B}} =0.
\end{equation*}
This yields to \eqref{Gaug276}.
\end{proof}

\begin{proof}[\bf Proof of Theorem \ref{mainthm1}]
Under the assumption of Theorem \ref{mainthm1}, $\sigma_{ess}(H)=\sigma_{ess}(H_0)$ \cite{ls,Gun91}.
This yields that $ \sigma_{ac}(H)\subset S$.
It is well known that boundedness of  the eigensolution  implies  purely absolutely continuous spectrum (e.g. \cite{bbound,stolz2}). Thus the second part of Theorem \ref{mainthm1} implies the first part.
If $p=1$ ($V\in L^1(\R^+)$), one has that for every $E\in K$, $\frac{iV }{2\Re\phi^{\prime}}e^{ih}$ given by \eqref{Gy} is in $L^1$.
In this case, it is well known (see \cite{Dollard1978} for example) that \eqref{Gy} has a solution  $Y(x)$ satisfying
\begin{equation*}
  Y(x)=\left(\begin{array}{c}
         1 \\
         0
       \end{array}\right)+o(1),
\end{equation*}
as $x\to \infty$. By Lemma \eqref{le10},  Theorem \ref{mainthm1} is true for $p=1$. So we assume  $1<p<2$.

 By Corollary \ref{Cor1} and  Theorem \ref{prop41},
 for almost every $E\in K$ the following limit is well defined,
$$T_{2m}(\mathcal{F})(x,\infty,E)=\lim_{x'\to \infty}T_{2m}(\mathcal{F})(x,x',E) .$$
By \eqref{Gaug234} and \eqref{Gaug277}, we have
\begin{eqnarray}
  |T_{2m}(\mathcal{F})(x,\infty,E)| &\leq& \frac{ C(E)^{2m}}{\sqrt{ (2m)!}}.\label{Gaug246}
\end{eqnarray}
Thus
\begin{equation}\label{Gaug247}
   \sum_{m=1}^{\infty}T_{2m}(\mathcal{F})(x,\infty,E)
\end{equation}
is absolutely convergent  for almost every $E\in K$.
Similarly,
\begin{equation}\label{Gaug248}
\sum_{m=0}^{\infty}T_{2m+1}(\mathcal{F})(x,\infty,E)
\end{equation}
is  absolutely convergent  for almost every $E\in K$.

 \eqref{Gaug247} and \eqref{Gaug248}
show that the right hand series in \eqref{Gseri1} and \eqref{Gseri2} are well defined for almost every $E\in  K$.
Based on \eqref{Gaug256new},
it is easy to check  for  almost every $E\in  K$, the right hand series in \eqref{Gseri1} and  \eqref{Gseri2}  actually  gives a solution of \eqref{Gy}.
The WKB behavior \eqref{Gwkb} follows by \eqref{Gaug256} and Lemma \ref{le10}.
\end{proof}

\section{An alternative proof of Theorem  \ref{mainthm1} and proof of Theorem \ref{mainthm2} }\label{Newproof}
We will give a new proof of Theorem \ref{mainthm1}. As the arguments in the previous section, it suffices to prove Corollary \ref{Cor1}.
The proof of Corollary \ref{Cor1} in the previous section is based on the  maximal operator $S^{\ast}$. We will give a new proof without using  the  maximal operator.

\begin{lemma}\cite{mkjfa2}
Let  $\{E_j^m\subset \R_+:m\in \Z_+,  1\leq j\leq 2^m\}$  be a martingale structure.
Then there exists an absolute constant $C$ such that for any closed interval $I$
\begin{equation*}
  ||g\chi_{I}||_{\mathfrak{B}^s}\leq C  ||g||_{\mathfrak{B}^{s+1}}.
\end{equation*}
In particular,
\begin{equation}\label{Gaug278}
  ||g\chi_{I}||_{\mathfrak{B}}\leq C  ||g||_{\mathfrak{B}^{2}}.
\end{equation}
\end{lemma}
\begin{proof}[\bf A new proof of Corollary  \ref{Cor1} without using $S^{\ast}$]
The   proof of \eqref{Gaug276} does not use $S^{\ast}$. So we keep it.
We only need to show \eqref{Gaug277} is true for almost every $E\in K$.
Applying $P=S$ and $ f=V(x)$ in Theorem \ref{thmbound} and recalling that
\begin{equation*}
 ||\mathcal{F}(\cdot,E)||_{\mathfrak{B}^2}=G^{(2)}_{S(V)(E)},
\end{equation*}
one has
\begin{equation}\label{Gaug279}
   ||\mathcal{F}(\cdot,E)||_{\mathfrak{B}^2}\in L^q(J).
\end{equation}
Now \eqref{Gaug277} follows  from \eqref{Gaug278} and \eqref{Gaug279}.
\end{proof}
Suppose  the assumptions of Theorem \ref{mainthm2} hold for some $1<p\leq 2$ and $\gamma>0$.
 Let $\beta$ be any positive number bigger than $1-p^{\prime}\gamma$. Denote by $  \mathcal{H}^{\beta}$ the $\beta$-dimensional Hausdorff measure.
Let
\begin{equation*}
  \Lambda_c=\{E\in K: ||\mathcal{F}(\cdot,E)\chi_{[N,\infty)}||_{\mathfrak{B}^2}\geq c \text{ for every } N\geq 0\}.
\end{equation*}
\begin{lemma}\label{Lehaus}
For any $c>0$,
we have
\begin{equation*}
  \mathcal{H}^{\beta}(\Lambda_c)=0.
\end{equation*}
\end{lemma}
\begin{proof}
The Lemma follows  from the  arguments   in \cite{mkcmp01}. Actually, the proof of Lemma \ref{Lehaus} is easier since the potential $V$ does not depend on $E$.
\end{proof}
\begin{proof}[\bf Proof of Theorem \ref{mainthm2}]
By Lemma \ref{Lehaus},
 we have for all $E$ except for a set of $\mathcal{H}^{\beta}$  zero,  for any $c>0$, there exists $N>0$ such that
\begin{equation}\label{Gau10}
||\mathcal{F}(\cdot,E)\chi_{[N,\infty)}||_{\mathfrak{B}^2} \leq c .
\end{equation}
Fix such $E$.
Let $N_0$ be  such that  \eqref{Gau10} holds for  $c=1$. By changing $x$ to $x-N_0$, we can assume $N_0=0$.
Hence, by \eqref{Gaug278}, one has
\begin{equation*}
  \sup_{I\subset \R_+}||\mathcal{F}(\cdot,E)\chi_{I}||_{\mathfrak{B}}\leq C .
\end{equation*}

For any $\varepsilon>0$, let $N(\varepsilon)$ be large enough so that \eqref{Gau10} holds for $c=\varepsilon$.
For any $M>N(\varepsilon)$, by \eqref{Gaug278} again,
\begin{eqnarray*}
  ||\mathcal{F}(\cdot,E)\chi_{[M,\infty)}||_{\mathfrak{B}} &\leq&  ||\mathcal{F}(\cdot,E)\chi_{[N(\varepsilon),\infty)}||_{\mathfrak{B}}+||\mathcal{F}(\cdot,E)\chi_{[N(\varepsilon),M)}||_{\mathfrak{B}} \\
 &\leq& ||\mathcal{F}(\cdot,E)\chi_{[N(\varepsilon),\infty)}||_{\mathfrak{B}}+C||\mathcal{F}(\cdot,E)\chi_{[N(\varepsilon),\infty)}||_{\mathfrak{B}^2}\\
 &\leq& C\varepsilon.
\end{eqnarray*}
This implies \eqref{Gaug276}.  Now the rest proof of Theorem \ref{mainthm2} follows from that of  Theorem \ref{mainthm1}.
\end{proof}
\section{Sharp estimates for  almost orthogonality among  generalized  Pr\"ufer angles}
In this section, we always assume for some $B>0$,
\begin{equation}\label{GVbound}
 |V(x)|\leq \frac{B}{1+x}.
\end{equation}
Without loss generality, we only consider the Dirichlet boundary condition.

For any spectral band $[a_n,b_n]$, let $c_n$ be the unique number such that $k(c_n)=\frac{\pi}{2}$.
Let $I$ be a closed interval in $(a_n,c_n)$ or $(c_n,b_n)$. All the energies $E$ in this section are in $I$ and the estimates are uniform with respect to $E\in I.$

For $z\in \C\backslash \R$, denote by $\tilde{v}(x,z)$ ($\tilde{u}(x,z)$) the solution of  $H_0+V$ with boundary condition $\tilde{v}(0,z)=1$ and $\tilde{v}^{\prime}(0,z)=0$ ($\tilde{u}(0,z)=0$ and $\tilde{u}^{\prime}(0,z)=1$).
The Weyl $m$-function $m(z)$ (well defined on $z\in \C\backslash \R$) is given by the unique complex number $m(z)$ so that  $\tilde{v}(x,z)+m(z)\tilde{u}(x,z)\in L^2(\R^+)$.
The spectral measure $\mu$ on $\R$, is given by the following formula, for $z\in \C\backslash \R$
\begin{equation*}
  m(z)=C+\int\left[\frac{1}{x-z}-\frac{x}{1+x^2}\right]d\mu(x),
\end{equation*}
where $ C\in \R$ is a constant.

Denote $\mu_{sc}$  by the singular continuous   component of $\mu$.
The spectral measure $\mu$ can determine the spectra of $H$. For example,
 $\sigma_{\rm sc}(H_0+V)=\emptyset $ if and only if $\mu_{sc}=0$.

Define  $\gamma(x,E)$  as a continuous function such that
\begin{equation}\label{Gfl}
  \varphi(x,E) =|\varphi(x,E)|e^{i \gamma(x,E)}.
\end{equation}

By \cite[Proposition 2.1]{KRS}, we know there exists some constant  $C>0$   such that
\begin{equation}\label{Ggammad}
   \frac{1}{C}\leq \gamma^\prime(x,E)\leq C.
\end{equation}

Define $\rho(x,E)\in \C$ by
\begin{equation}\label{Grho}
  \left(\begin{array}{c}
     u(x,E) \\
     u^{\prime}(x,E)
   \end{array}
   \right)=\frac{1}{2i}\left[\rho(x,E) \left(\begin{array}{c}
     \varphi(x,E) \\
     \varphi^{\prime}(x,E)
   \end{array}
   \right)-\overline{\rho}(x,E) \left(\begin{array}{c}
    \overline{ \varphi}(x,E) \\
     \overline{\varphi}^{\prime}(x,E)
   \end{array}
   \right)\right].
\end{equation}
We should mention that $u(x,E)$ is a real function and  $\varphi(x,E)$ is not.

Define $R(x,E)$ and $\theta(x,E)$ by
\begin{equation}\label{GR1}
R(x,E)=|\rho(x,E)|; \;\;\theta(x,E)=\gamma(x,E)+{\rm Arg}(\rho(x,E)).
\end{equation}

\begin{proposition}\cite{KRS}\label{thmformula}
Suppose $u$ is a real solution of \eqref{Gu}.
Then there exist real functions $ R(x)>0$ and $\theta(x)$ such that
\begin{equation}\label{GformulaR}
    [\ln R(x,E)]^\prime=\frac{V(x)}{2\gamma^\prime(x,E)} \sin2 \theta(x,E)
\end{equation}
and
\begin{equation}\label{Gformulatheta}
   \theta(x,E)^\prime=\gamma^\prime(x,E)-\frac{V(x)}{2\gamma^\prime(x,E)} \sin^2 \theta(x,E).
\end{equation}


\end{proposition}
In the following, we will  prove the almost orthogonality among Pr\"ufer angles in a sharp way.
Before the proof, some preparations are necessary.
\begin{lemma}\cite{Kiselev05}\label{Leper1}
Suppose the function $G(x)$ satisfies $|G^{\prime}(x)|=\frac{O(1)}{1+x}$ and $\gamma\neq 0$. Then
\begin{equation*}
    \left|\int_{0}^L\frac{\sin(\gamma x+G(x))}{1+x}dx\right|\leq O(1)\log|\gamma^{-1}|+O(1).
\end{equation*}

\end{lemma}
\begin{lemma}\cite{ld}\label{Leper2}
Suppose the function $G(x)$ satisfies $|G^{\prime}(x)|=\frac{O(1)}{1+x}$ and $\gamma\neq 0$. Then
\begin{equation*}
    \left|\int_{0}^L\frac{\sin(\gamma x+G(x))}{1+x}dx\right|\leq \frac{O(1)}{|\gamma|} +O(1).
\end{equation*}

\end{lemma}
\begin{lemma}\label{Leper4}
Suppose $0<\gamma<2\pi$ and the function $G(x)$ satisfies $|G^{\prime}(x)|=\frac{O(1)}{1+x}$. Then
\begin{equation*}
    \left|\int_{0}^Le^{2\pi i kx}\frac{\sin(\gamma x+G(x))}{1+x}dx\right|\leq O(1)\log  \gamma^{-1}+O(1)\log(2\pi -\gamma)^{-1} +O(1),
\end{equation*}
for $k=-1,0,1$,
and
\begin{equation*}
    \left|\int_{0}^Le^{2\pi i kx}\frac{\sin(\gamma x+G(x))}{1+x}dx\right|\leq \frac{O(1)}{|k|} +O(1),
\end{equation*}
for $k\in \Z\backslash \{-1,0,1\}$.
\end{lemma}
\begin{proof}
The proof follows from Lemmas \ref{Leper1}, \ref{Leper2} and trigonometric identity.
\end{proof}
Denote by $\T=\R/\Z.$
\begin{theorem}\label{Leper3}
Suppose $f\in L^2(\T)$. 
Then the following estimates hold
\begin{equation}\label{Gcons4}
  \left|  \int _{0}^L f(x)\frac{\cos 4 \theta(x,E)}{1+x}dx\right |=  O(1),
\end{equation}
and
\begin{equation}\label{Gcons5}
    \left| \int_{0}^L f(x)\frac{\sin 2 \theta(x,E_1) \sin 2 \theta(x,E_2)}{1+x}dx\right |=O(1)\log\frac{1}{|E_1-E_2|}+O(1),
\end{equation}
where $O(1)$ only depends on $ I $, $B$, $f$ and $V_0$.
\end{theorem}
\begin{proof}
We only give the proof of \eqref{Gcons5}.
 First, using \eqref{Gformulatheta} and \eqref{GVbound} we have the differential equations of $\theta(x,E_1)$ and $\theta(x,E_2)$,
\begin{equation}\label{Gfequ1}
\theta^\prime(x,E_1)=\gamma^\prime(x,E_1)+\frac{O(1)}{1+x},
\end{equation}
and
\begin{equation}\label{Gfequ2}
\theta^\prime(x,E_2)=\gamma^\prime(x,E_2)+\frac{O(1)}{1+x}.
\end{equation}
By \eqref{Gdefj} and \eqref{Gfl}, we have
\begin{equation}\label{Gfequ3}
    \gamma (x,E)=k(E)x+\eta(x,E),
\end{equation}
where $\eta(x,E)\mod 2\pi$ is a function that is $1$-periodic in $x$.

By basic trigonometry,
\begin{equation}
-2\sin 2\theta(x,{E}_1)\sin 2\theta(x, {E}_2)=\cos (2\theta(x,{E}_1)+2\theta(x,{E}_2))-\cos (2\theta(x,{E}_1)-2\theta(x,{E}_2)),
\end{equation}
 it suffices  to bound
\begin{equation*}
\int_{0}^{L}f(x)\frac{\cos (2\theta(x,{E}_1)\pm2\theta(x,{E}_2))}{  1+x} dx.
\end{equation*}
Without of loss of generality,  we only  bound
\begin{equation*}
\int_{0}^{L}f(x)\frac{\cos (2\theta(x,{E}_1)-\theta(x,{E}_2))}{  1+x} dx.
\end{equation*}
By \eqref{Gfequ1}, \eqref{Gfequ2} and \eqref{Gfequ3}, we have
\begin{equation}\label{GfourierTheta}
\dfrac{d}{d x}([\theta(x,{E}_1)-\eta(x,E_1)]-[\theta(x,{ E_2})-\eta(x,{E}_2)]=k(E_1)-k({E}_2)+\frac{O(1)}{1+x}.
\end{equation}
Let
\begin{equation*}
  \tilde{ {\theta}}(x,E) =\theta(x,{E})-\eta(x,E).
\end{equation*}
 By trigonometry again,
 one has
 \begin{eqnarray*}
   \cos (2\theta(x,{E}_1)-2\theta(x,{E}_2)) &=& \cos (2\tilde{\theta}(x,{E}_1)-2\tilde{\theta}(x,{E}_2)+2\eta(x,E_1)-2\eta(x,{E}_2)) \\
    &=& \cos(2\eta(x,E_1)-2\eta(x,{E}_2))\cos (2\tilde{\theta}(x,{E}_1)-2\tilde{\theta}(x,{E}_2))\\
    &&-\sin(2\eta(x,E_1)-2\eta(x,{E}_2))\sin (2\tilde{\theta}(x,{E}_1)-2\tilde{\theta}(x,{E}_2)).
 \end{eqnarray*}
  Thus
  $$\int_{0}^L f(x)\frac{\cos (2\theta(x,{E}_1)-2\theta(x,{E}_2))}{ 1+x} dx\;\;\;\;\;\;\;\;\;\;\;\;\;\;\;\;\;\;\;\;\;\;\;\;\;\;\;\;\;\;\;\;\;\;\;\;\;\;\;\;\;\;\;\;\;\;\;\;$$
\begin{eqnarray*}
  \;\;\;\;\;\;\;\;\;\;\;\;\;\;\;\;\;\;\;\;\;\;\;\;\;\;\;\;\;\;\;\; &=& \int_{0}^L f(x)\frac{\cos(2\eta(x,E_1)-2\eta(x,{E}_2))\cos (2\tilde{\theta}(x,{E}_1)-2\tilde{\theta}(x,{E}_2))}{ 1+x} dx \\
 &-&\int_{0}^L f(x)\frac{\sin(2\eta(x,E_1)-2\eta(x,{E}_2))\sin (2\tilde{\theta}(x,{E}_1)-2\tilde{\theta}(x,{E}_2))}{ 1+x}dx.
\end{eqnarray*}
Without loss of generality, we only give the estimate for
\begin{equation}\label{Gfourier5}
\int_{0}^L f(x)\frac{\sin(2\eta(x,E_1)-2\eta(x,{E}_2))\sin (2\tilde{\theta}(x,{E}_1)-2\tilde{\theta}(x,{E}_2))}{ 1+x}dx.
\end{equation}
We proceed by Fourier expansion of $f(x)\sin(2\eta(x,E_1)-2\eta(x,{E}_2))$ ($1$-periodic function) and obtain that
\begin{equation*}
f(x)\sin(2\eta(x,E_1)-2\eta(x,{E}_2))=\frac{c_0}{2}+\sum_{k=1}^\infty c_k \cos(2\pi kx)+ d_k \sin(2\pi kx).
\end{equation*}
Plugging this back into \eqref{Gfourier5}, we obtain
\begin{align}\nonumber\eqref{Gfourier5}=\int_{0}^L \frac{c_0}{2} \frac{\sin (2\tilde{\theta}(x,{E}_1)-2\tilde{\theta}(x,{E}_2))}{1+x}dx+\sum_{k=1}^\infty c_k \cos(2\pi kx)\frac{\sin (2\tilde{\theta}(x,{E}_1)-2\tilde{\theta}(x,{E}_2))}{1+x}dx\\+\sum_{k=1}^\infty d_k \sin(2\pi kx)\frac{\sin (2\tilde{\theta}(x,{E}_1)-2\tilde{\theta}(x,{E}_2))}{1+x}dx.\label{eq:Fouriersum}
\end{align}

Since   $k(E_1), k( E_2)\in(0,\frac{\pi}{2})$ or  $k(E_1), k( E_2)\in(\frac{\pi}{2},\pi)$  depending on either $I\subset(a_n,c_n)$ or $I\subset (c_n,b_n)$, and $k(E_1)\neq k(E_2)$, we  have 
\begin{equation}\label{Gfourier6}
0<|k(E_1)- k(  E_2)|< \frac{\pi}{2}.
\end{equation}
Since $f\in L^2(\T)$, on has  $\sum c_k^2+d_k^2<\infty$.
 Now the Theorem  follows from Lemma \ref{Leper4}, \eqref{eq:Fouriersum} and \eqref{Gfourier6}.

\end{proof}

\section{ Spectral analysis  of Schr\"odinger operators with eventually periodic potentials}\label{even}

For  $L>0$, let $ V_L$ be the cut off $V$ up to $L$. More precisely, $V_L(x)=V(x)$ for $0\leq x\leq L$ and $V_L(x)=0$ for $x>L$.
Let $ \mu_L$ be the spectral measure corresponding to the operator with potential $V_L$.
\begin{theorem}\label{Lemu}
The following formula hold,
\begin{equation}\label{Gmul}
   \frac{d\mu_L(E)}{d E}= \frac{2}{\pi|W( \overline{\varphi},\varphi)|}\frac{1}{R^2(L,E)}
\end{equation}
for $E\in S$.

\end{theorem}
\begin{proof}
Let $Q(z)$ be the transfer matrix of $H_0 $ from $0$ to $1$, that is
\begin{equation*}
  Q (z)\left(\begin{array}{cc}
               \phi(0) \\ \phi^{\prime} (0)
              \end{array}
  \right)=\left(\begin{array}{cc}
                \phi(1) \\ \phi^{\prime}(1)
              \end{array}
  \right)
\end{equation*}
for any solution $\phi$ of $ -\phi^{\prime\prime}+V_0\phi=z\phi$.

Let $z=E+i\varepsilon$ for $E\in S$ and $\varepsilon\geq0$. Let $ k(z)+i \tau(z)$ be such that $2\cos (k(z)+i\tau(z))={\rm Tr } Q(z)$ with $ k(z)\in \R$ and $\tau(z)\in \R $.  Thus
\begin{equation*}
  (  e^{-\pi \tau}+e^{\pi \tau})\cos k=\Re {\rm Tr} Q(z); (  e^{-\pi \tau}-e^{\pi \tau})\sin k= \Im  {\rm Tr} Q(z).
\end{equation*}
Let us choose the branch  so that $k(z)\in(0,\pi)$.
Thus
\begin{equation*}
    \lim_{\varepsilon\to 0+}k(E+i\varepsilon)=k(E),    \lim_{\varepsilon\to 0+}\tau(E+i\varepsilon)=0
\end{equation*}
where $k(E)$ is  quasimomentum of $E\in S$.

By the following Lemma \ref{Lepositivegamma}, we have either $ \tau(E+i\varepsilon)>0$ for all $\varepsilon>0$ and $E\in S$
or $ \tau(E+i\varepsilon)<0$ for all $\varepsilon>0$ and $E\in S$. Without loss of generality, assume $ \tau(E+i\varepsilon)<0$.

By the Floquet theory, $-u^{\prime\prime}+V_0u=zu$ has a solution with the form:
\begin{equation}\label{Gvarphifli}
  \varphi(x,z)=J(x,z)e^{-i(k(z) +i\tau(z))x},
\end{equation}
where $J(x,z)$ is 1-periodic\footnote{We will show that $ \varphi(x,z)$ is an extension of  $\varphi(x,E)$ given by \eqref{Gdefj} and \eqref{Gflight1} later. This means that there is no confusion even through we use the same notation.}.

Define $ \tilde{u}(x,z)=J(x,z)e^{-i (k+i\tau)x}= \varphi(x,z)$ for $x\geq L$ and extend $\tilde{u}(x,z)$ to $0\leq x\leq L$  by solving equation
\begin{equation*}
  -\tilde{u}^{\prime\prime}(x,z)+ (V_0(x)+V_L(x)-z)\tilde{u}(x,z)=0
\end{equation*}
for $0\leq x\leq L$.  Since $ \gamma(z)<0$, one has  $\tilde{u}(x,z)\in L^2(\R^+)$.
By  spectral theory (we refer the readers to \cite{Simon} and references therein for details), we have
\begin{equation*}
    m(z)=\frac{\tilde{u}^{\prime}(0,z)}{\tilde{u}(0,z)},
\end{equation*}
and
\begin{equation}\label{Gequ13}
    \frac{d\mu_L}{dE}=\frac{1}{\pi}\lim_{\varepsilon\to 0+} \Im m(E+i\varepsilon).
\end{equation}

Let $T(z)$ be the transfer matrix of $H_0+V_L$ from $0$ to $L$, that is
\begin{equation*}
  T (z)\left(\begin{array}{cc}
               \phi(0) \\ \phi ^\prime(0)
              \end{array}
  \right)=\left(\begin{array}{cc}
                \phi(L) \\ \phi^{\prime}(L)
              \end{array}
  \right)
\end{equation*}
for any solution $\phi$ of $(-D^2+V_0+V_L)\phi=z\phi$.

Denote by
\begin{equation*}
  T(z)=\left(
           \begin{array}{cc}
             a(z) & b(z) \\
             c(z) & d(z) \\
           \end{array}
         \right).
\end{equation*}
Clearly,
\begin{eqnarray*}
  \left(\begin{array}{cc}
                \tilde{u}(0,z) \\ \tilde{u}^\prime (0,z)
              \end{array}
  \right) &=&\left(
           \begin{array}{cc}
             a(z) & b(z) \\
             c(z) & d(z) \\
           \end{array}
         \right)^{-1}\left(\begin{array}{cc}
                \tilde{u}(L,z) \\ \tilde{u} ^{\prime}(L,z)
              \end{array}
  \right) \\
   &=& \left(
           \begin{array}{cc}
             d(z) & -b(z) \\
             -c(z) & a(z) \\
           \end{array}
         \right) \left(\begin{array}{cc}
                \tilde{u}(L,z) \\ \tilde{u} ^\prime(L,z)
              \end{array}
  \right).
\end{eqnarray*}
Direct computation implies that (using $ad-bc=1$)
\begin{eqnarray}
 \nonumber   \lim_{\varepsilon\to 0+} \Im m(E+i\varepsilon) &=& \Im \frac{ a  \varphi^{\prime}(L,E)-c\varphi(L,E) }{d  \varphi(L,E)-b\varphi^{\prime}(L,E)} \\
   &=& \frac{A_1\sin B_1}{(d-bA_1\cos B_1)^2+(bA_1\sin B_1)^2},\label{Gequ12}
\end{eqnarray}
where $A_1>0$ and $B_1$ are defined by
\begin{equation}\label{GvarphiL}
  \frac{\varphi^{\prime}(L)}{\varphi(L)}=A_1e^{i B_1}.
\end{equation}
Since $\Im  m(E+i\varepsilon)\geq0$ for $\varepsilon\geq0$, using \eqref{Gequ12} and \eqref{GvarphiL}, one has
\begin{equation*}
  \Im (\overline{\varphi(x,E)}\varphi^{\prime}(x,E))\geq 0.
\end{equation*}
It implies that the restriction of  $\varphi(x,z)$ (given by \eqref{Gvarphifli}) on $z\in\R$  coincides with   $\varphi(x,E)$ given by \eqref{Gdefj} and \eqref{Gflight1}.

It is easy to see that
\begin{eqnarray}
 \nonumber \left(\begin{array}{cc}
                u(L,E) \\ u ^{\prime}(L,E)
              \end{array}
  \right) &=&T(E)\left(\begin{array}{cc}
                u(0) \\ u^{\prime}(0)
              \end{array}
  \right) \\
   &=& T(E)\left(\begin{array}{cc}
               0\\ 1
              \end{array}
  \right)= \left(\begin{array}{cc}
                b \\ d
              \end{array}
  \right).\label{Gbd}
\end{eqnarray}
Let
\begin{equation}\label{GA2B2}
  \rho(L)\varphi(L)=A_2e^{iB_2}\text{ and } A_2>0.
\end{equation}
By \eqref{Grho} and \eqref{Gbd}, we have
\begin{eqnarray}
 \nonumber \left(\begin{array}{cc}
                b \\ d
              \end{array}
  \right) &=&\Im A_2e^{iB_2}\left(\begin{array}{cc}
                A_1e^{iB_1} \\ 1
              \end{array}
  \right) \\
   &=&   \left(\begin{array}{cc}
            A_2\sin B_2     \\ A_1A_2\sin(B_1+B_2)
              \end{array}
  \right).\label{GB1B2}
\end{eqnarray}
Thus
\begin{equation}\label{Gb1b2}
  b= A_2\sin B_2;\;\; d= A_1A_2\sin(B_1+B_2) .
\end{equation}
By \eqref{Gequ12} and \eqref{Gb1b2},
\begin{equation}\label{Gequ14}
   \lim_{\varepsilon\to 0+} \Im m(E+i\varepsilon) = \frac{1}{A_1A_2^2\sin B_1}.
\end{equation}
It is easy to see that (see p.295 in \cite{KRS})
\begin{equation}\label{Gomega}
 |\varphi|^2 \Im\frac{\varphi^{\prime}}{\varphi}=\frac{|W( \overline{\varphi},\varphi)|}{2}.
\end{equation}
By \eqref{GvarphiL}, \eqref{GA2B2}, \eqref{Gequ14} and \eqref{Gomega},
\begin{eqnarray}
 \nonumber \lim_{\varepsilon\to 0+} \Im m(E+i\varepsilon) &=& \frac{1}{|\rho(L)\varphi(L)|^2\Im(\frac{\varphi^{\prime}(L)}{\varphi(L)})} \\
    &=& \frac{2}{|W( \overline{\varphi},\varphi)|} \frac{1}{R(L,E)^2}.\label{Gimm}
\end{eqnarray}
Now the Theorem follows from \eqref{Gequ13} and \eqref{Gimm}.
\end{proof}
\begin{lemma}\label{Lepositivegamma}
 Either $ \gamma(E+i\varepsilon)>0$ for all $\varepsilon>0$ and $E\in S$
or $ \gamma(E+i\varepsilon)<0$ for all $\varepsilon>0$ and $E\in S$.
\end{lemma}
\begin{proof}
Otherwise, there exist $E\in S$ and $\varepsilon$ such that $ \gamma(E+i\varepsilon)=0$. Let $z=E+i\varepsilon$.
Therefore,   the equation $ -u^{\prime\prime}+V_0u=zu$  has two linearly independent solutions $\varphi_1(x,z)$ and $\varphi_2(x,z)$,  with the form of
\begin{equation*}
  \varphi_1(x,z)=J_1(x,z)e^{-ik(z) x}\text { and }\varphi_2(x,z)=J_2(x,z)e^{ik(z) x},
\end{equation*}
where $J_1(x,z)$  and $J_2(x,z)$ are 1-periodic.
It implies $ -u^{\prime\prime}+V_0u=zu$ does not have   $L^2(\R^+)$ solution. This is impossible since $\Im z=\varepsilon>0$.
\end{proof}

\section{ Proof of Theorem \ref{mainthm3}}\label{lsection}
 In this section, we give the dependence of   parameters  explicitly except  $V_0$, since $V_0$ is fixed all the time.

 Define
 \begin{equation*}
  \Gamma(E)=\int_0^1 \frac{1}{|\gamma^{\prime}(x,E)|^2}dx.
\end{equation*}
It easy to see (See \cite[Prop.2.1]{KRS})
\begin{equation}\label{Ggamma}
  \gamma^{\prime}(x,E)=\frac{|W(\overline{\varphi},\varphi)|}{2|\varphi(x,E)|^2}.
\end{equation}
 Let $L= \epsilon^{-1-\sigma} $.
Let $C_1=C_1(B, I)$, which will be determined later.

We say a subset $ A\subset I$ is $(\epsilon,N)$ separate, if the following two conditions hold:

For any $E\in A$,
\begin{equation}\label{Gassum1}
    \left|\int_{0}^LV(x)\frac{\sin2\theta(x,E)}{\gamma^{\prime}(x,E)}dx\right|\geq (1-\beta)C_1(B,I)\log \epsilon^{-1}.
\end{equation}
For any $E_1,E_2\in A$ and $E_1\neq E_2$,
\begin{equation}\label{Gassum2}
   |k(E_1)-k(E_2)|\geq \epsilon^{1/N^2}.
\end{equation}
\begin{lemma}\cite[Lemma 4.4]{KLS}\label{Lee}
  Let $\{e_i\}_{i=1}^N$ be a set of unit vectors in a Hilbert space $\mathcal{H}$  so that
  \begin{equation*}
    \alpha=N\sup_{i\neq j}| \langle e_i,e_j\rangle|<1.
  \end{equation*}
  Then for any $g\in \mathcal{H}$,
  \begin{equation}\label{Gapr71}
    \sum_{i=1}^N|\langle g,e_i\rangle|^2\leq (1+\alpha)||g||^2.
  \end{equation}
  \end{lemma}
\begin{theorem}\label{Thmbound}
There exists $\epsilon_1(B,I,\sigma,\beta)>0$ and $C(B,I,\sigma,\beta)$ such that  for any $\epsilon<\epsilon_1$ and $N\geq C(B,I,\sigma,\beta)$,  the
$(\epsilon,N)$ separate  set $A$ satisfies $\# A\leq N$.
\end{theorem}
\begin{proof}
We consider the Hilbert space
\begin{equation*}
  \mathcal{H}=L^2((0,L),(1+x)dx).
\end{equation*}

In  $\mathcal{H}$, by \eqref{GVbound} we have
\begin{equation}\label{Gapr77}
  ||V||_{ \mathcal{H}}^2\leq B^2\log (1+L).
\end{equation}
Let
\begin{equation}\label{Gdefei}
  e_{i}(x)=\frac{1}{\sqrt{A_i}}\frac{ \sin 2\theta(x,E_i)}{\gamma^{\prime}(x,E_i)(1+x)}\chi_{[0,L]}(x),
\end{equation}
where $A_i$ is chosen so that $e_i$ is an unit vector in $\mathcal{H}$.
We have the following estimate,
\begin{eqnarray}
  A_i &=& \int_{0}^{B_j}\frac{\sin^2 2\theta(x,E_i)}{|\gamma^{\prime}(x,E_i)|^2(1+x)}dx \nonumber\\
   &=&\int_{0}^{L}\frac{1}{2|\gamma^{\prime}(x,E_i)|^2(1+x)}dx - \int_{0}^{L}\frac{\cos 4\theta(x,E_i)}{|\gamma^{\prime}(x,E_i)|^2(1+x)}dx \label{Gapr791}.
\end{eqnarray}

By \eqref{Gcons4}, one has
\begin{equation}\label{Gcons4Gapr79}
  \left| \int_{0}^{L}\frac{\cos 4\theta(x,E_i)}{|\gamma^{\prime}(x,E_i)|^2(1+x)}dx\right|= O(1).
\end{equation}

Direct computation shows that
\begin{eqnarray}
  \int_{0}^{L}\frac{1}{|\gamma^{\prime}(x,E_i)|^2(1+x)}dx &=& O(1)+\sum_{n=0}^{L-1} \int_{n}^{n+1}\frac{1}{|\gamma^{\prime}(x,E_i)|^2(1+n)}dx\nonumber\\
  &=& O(1)+\Gamma(E_i)\log L .\label{Gamma1}
\end{eqnarray}
By \eqref{Gapr791}, \eqref{Gcons4Gapr79} and \eqref{Gamma1}, we have
\begin{equation}\label{Gapr79}
      A_i= \frac{1}{2}\Gamma(E_i)\log L +O(1).
\end{equation}
We should mention that $O(1)$ in \eqref{Gcons4Gapr79}, \eqref{Gamma1} and \eqref{Gapr79} only depends on $B$ and $I$.

By \eqref{Gcons5} and \eqref{Gdefei}, we have
\begin{equation}\label{Gapr78}
 | \langle e_i ,e_{ j}  \rangle\leq \frac{2}{1+\sigma}C(I,B)N^{-2} +\frac{C(I,B)}{\log \epsilon^{-1}}.
\end{equation}
The first condition  \eqref{Gassum1} implies
\begin{equation}\label{Gequ5}
    |\langle V,e_i \rangle|^2\geq \frac{C_1^2}{1+\sigma}\log \epsilon^{-1}-C(I,B).
\end{equation}
By \eqref{Gapr71} and \eqref{Gapr78}, one has
\begin{equation}\label{Gapr711}
\sum_{i=1}^N  |\langle V,e_i\rangle_{\mathcal{H}}|^2\leq (1+\frac{2}{1+\sigma}C(I,B)N^{-1} +\frac{NC(I,B)}{\log \epsilon^{-1}})||V||_{\mathcal{H}}.
\end{equation}
By \eqref{Gequ5}, \eqref{Gapr77} and \eqref{Gapr711}, we have
\begin{equation*}
    N\left(  \frac{C_1^2(1-\beta)^2}{1+\sigma}\log \epsilon^{-1}-C(I,B)\right)\leq  \left(1+\frac{2}{1+\sigma}C(I,B)N^{-1} +\frac{NC(I,B)}{\log \epsilon^{-1}}\right) B^2(1+\sigma)\log \epsilon^{-1}.
\end{equation*}
This implies the Theorem.
\end{proof}
\begin{proof}[\bf Proof of Theorem \ref{mainthm3}]
Once we have Theorems \ref{Lemu} and \ref{Thmbound}, Theorem \ref{mainthm3} can be proved by the arguments in \cite{Kiselev05} (also see \cite{liudis}). We omit the details here.
\end{proof}
\appendix
\section{}
\begin{proof}[\bf Proof of   \eqref{Gy}]
By \eqref{Gu_3} and \eqref{Gy}, it suffices to show that
\begin{equation}\label{Gaug192}
   \frac{i}{2\Im (\varphi \overline{\varphi}')}V(x) \overline{\varphi}^{2}\exp (-2ip)=  \frac{-iV }{2\Re\phi^{\prime}}e^{-ih}
\end{equation}
By the definition, one has
\begin{equation*}
  \phi^{\prime}=-i\frac{\varphi'}{\varphi}.
\end{equation*}
Direct computations imply
\begin{equation}\label{Gaug191}
  \frac{|\varphi|^2}{\Im(\varphi \bar{\varphi}^{\prime})}=-\frac{1}{\Re \phi'}.
\end{equation}
By the definitions of $ h$ and $p$, we have
\begin{eqnarray*}
   \frac{iV }{2\Re\phi^{\prime}}e^{-ih}  
   &=&   \frac{iV}{2\Re\phi^{\prime}}\exp\left(-i2 \Re \phi+i\int_{0}^x\frac{V(t)}{\Re \phi^{\prime}(t,E)}dt \right)\\
    &=&   - iV\frac{|\varphi|^2}{2\Im(\varphi \bar{\varphi}^{\prime})} \exp(-i2 \Re \phi)\exp\left(-i\int_{0}^x\frac{V(t)|\varphi(t,E)|^2}{\Im(\varphi \bar{\varphi}^{\prime})}dt \right)\\
     &=&    -iV\frac{|\varphi|^2}{2\Im(\varphi \bar{\varphi}^{\prime})} \exp(-i2 \Re \phi) \exp(-2ip) \\
       &=& -   iV\frac{\bar{\varphi}^2}{2\Im(\varphi \bar{\varphi}^{\prime})} \exp(-2ip) .
\end{eqnarray*}
It implies \eqref{Gaug192} and hence \eqref{Gy}.

\end{proof}
\begin{proof}[\bf Proof of \eqref{Gaug195}]
Denote by
\begin{equation*}
  f_k=\int_{k-1}^{k} |f(x)|dx
\end{equation*}
for $k\in\Z_+$. Then
\begin{equation}\label{Gaug196}
  ||f||_{\ell^2(L^1)}^2=\sum_{k=1}^{\infty} f_k^2.
\end{equation}
Direct computations imply
\begin{eqnarray}
  \int_{\R_+^2}\frac{|f(x)f(y)|}{1+|x-y|^2}dxdy &=&O(1) \sum_{m=1}^{\infty} \sum_{n=1}^{\infty} \frac{f_m f_n}{1+|m-n|^2}\nonumber \\
  &=& O(1)\sum_{n=1}^{\infty} | f_n|^2, \label{Gaug197}
\end{eqnarray}
where the second equality holds by Young's convolution inequality.
Now \eqref{Gaug195}  follows from \eqref{Gaug196} and \eqref{Gaug197}.
\end{proof}
 \section*{Acknowledgments}
I would like to thank Barry Simon and Gunter Stolz for some useful discussion.
    This research was 
    supported by   NSF DMS-1700314/2015683, DMS-2000345 and  the Southeastern
    Conference (SEC) Faculty Travel Grant 2020-2021.

\end{document}